\documentclass[12pt]{article}

\usepackage{amsthm}
\usepackage{amssymb}
\usepackage{amsmath}
\usepackage{amsfonts}
\usepackage{times}

\usepackage[paper=a4paper, top=2.5cm, bottom=2.5cm, left=2.5cm, right=2.5cm]{geometry}

\def\qed{\hfill $\vcenter{\hrule height .3mm
\hbox {\vrule width .3mm height 2.1mm \kern 2mm \vrule width .3mm
height 2.1mm} \hrule height .3mm}$ \bigskip}

\def \RR {\mathbb R}

\def \EE {\mathbb E}

\def \Var {\mathrm{Var}}

\def \PP {\mathbb P}
\def \eps {\varepsilon}

\def \DC {\mathcal{C}_N}

\def \Id {\mathrm{I}_n}
\def \COV {\mathrm{Cov}}

\def \FF {\mathcal{F}}

\def \One {\vec{\mathbf{1}}}

\newcommand{\bsig}{\mathbf{\sigma}}
\def \gg {\mathbf{g}}
\def \LL { \mathcal{L} }
\def \JJ { \mathcal{J} }

\newtheorem{theorem}{Theorem}
\newtheorem{lemma}[theorem]{Lemma}

\newtheorem{corollary}[theorem]{Corollary}
\theoremstyle{definition}
\newtheorem{definition}[theorem]{Definition}
\theoremstyle{remark}
\newtheorem{remark}[theorem]{Remark}

\long\def\symbolfootnotetext[#1]#2{\begingroup
\def\thefootnote{\fnsymbol{footnote}}\footnotetext[#1]{#2}\endgroup}

\title{A simple approach to chaos for $p$-spin models}
\author{Ronen Eldan\thanks{Supported by a European Research Council Starting Grant (ERC StG) and by an Israel Science Foundation (ISF) Grant no. 715/16}}
\begin{document}

\maketitle
\begin{abstract}
We prove that, in mixed $p$-spin models of spin glasses, the location of the ground state is chaotic under small Gaussian perturbations. For the case of even $p$-spin models, this was shown by Chen, Handschy and Lerman \cite{Chen2018}. We rely on a different approach which only uses the Parisi formula as a black box.
\end{abstract}

\section{Introduction}
This paper concerns with the mixed $p$-spin model of spin-glasses, defined as follows. Fix a dimension $N \in \mathbb{N}$ and fix non-negative constants $(c_p)_{p=2}^\infty$ normalized so that $\sum_p c_p^2 = 1$. Define
$$
\xi(s) := \sum_{p=2}^\infty c_p^2 s^p.
$$
Consider the discrete hypercube $\DC = \{-1,1\}^N$. For $x \in \RR^N$, define
$$
\JJ(x) = \bigoplus_{p=2}^\infty \frac{c_p}{N^{(p-1)/2}} x^{\otimes p} \in \mathcal{H},
$$
where $\mathcal{H}$ is the Hilbert space $\bigoplus_{p=2}^\infty \left (\RR^N \right )^{\otimes p}$. 

The mixed $p$-spin model is a Gaussian process indexed by $\DC$, defined as $\bsig \to H_N(\bsig)$ with covariance structure
$$
\COV( H_N(\sigma^1), H_N(\sigma^2) ) = \left \langle \JJ(\sigma^1), \JJ(\sigma^2) \right \rangle = N \xi \left ( \frac{1}{N} \langle \sigma^1,\sigma^2 \rangle \right ), ~~ \forall \sigma^1,\sigma^2 \in \DC.
$$
Leting $\gg$ be a vector of independent, standard Gaussian entries in $\bigoplus_{p=2}^\infty \left (\RR^N \right )^{\otimes p}$, we will define this Gaussian process more explicitly by setting $H_N (\sigma) = H_N(\sigma, \gg)$ where
$$
H_N(\sigma; x) := \langle \JJ(\sigma), x \rangle_{\mathcal{H}}
$$
(using the scalar product notation is a slight abuse of notation since $\gg$ is not in the Hilbert space $\mathcal{H}$. However, since $\|\JJ(\sigma)\|_{\mathcal{H}} = 1$, the above is well-defined).

The special case of this model, corresponding to $c_p = \mathbf{1}_{\{p=2\}}$, is the so-called Sherrington-Kirkpatrick spin glass which can be equivalently defined as $H_N(\sigma) = \sum_{i,j} g_{i,j} \sigma_i \sigma_j$ where $(g_{i,j})_{1\leq i,j \leq N}$ are standard Gaussians.

Consider the ground state
$$
\sigma^*(x) := \arg \max_{\sigma \in \DC} H_N(\sigma; x).
$$
This paper is concerned with the question: \\ 

\emph{How stable is $\sigma^*(\gg)$ with respect to small perturbations of $\gg$?} \\

\noindent To make the question precise, let $\gg'$ be an independent copy of $\gg$ and for $t\geq 0$ set $\gg^t = e^{-t} \gg + \sqrt{1-e^{-2t}} \gg'$, so that
$$
(\gg, \gg^t) \sim \mathcal{N} \left (0, \left ( \begin{matrix} \mathrm{Id} & e^{-t} \mathrm{Id} \\ e^{-t} \mathrm{Id} & \mathrm{Id} \end{matrix} \right ) \right ).
$$ 
For small $t$, $\gg^t$ can be thought of as a noisy version of $\gg$. We are interested in the question of whether there exist a sequence $\eps_N \to 0$ such that
\begin{equation}\label{eq:chaos}
\lim_{N \to \infty} \EE \left [ \xi \left ( \frac{1}{N} \left \langle \sigma^*(\gg), \sigma^*(\gg^{\eps_N}) \right \rangle  \right ) \right ] = 0.
\end{equation}

A Gaussian process which satisfies \eqref{eq:chaos} is said to exhibit the \emph{chaos} property. While above question was essentially posed in the physics literature, the precise definition of chaos in the broader context of Gaussian fields was made in the seminal paper of Chatterjee \cite{chatterjee2008chaos}, where it is also shown that chaos is related to several other natural properties of Gaussian fields which have a disordered nature, some of which we discuss below. \\

\begin{remark}
In equation \eqref{eq:chaos} it may be natural to ask whether one can replace the expression $\xi \left ( \frac{1}{N} \left \langle \sigma^*(\gg), \sigma^*(\gg^{\eps_N}) \right \rangle  \right )$ by the expression $\left | \frac{1}{N} \left \langle \sigma^*(\gg), \sigma^*(\gg^{\eps_N}) \right \rangle  \right |$, which would correspond to a stronger bound in some cases (the two are equivalent for even $p$-spin models). We do not know whether or not such a strengthening is true, but we point out that our version of the bound corresponds to the definition of chaos which usually appears in the literature, in particular in \cite{chatterjee2008chaos}. 
\end{remark}

In the case of even $p$-spin models (hence when $c_p=0$ for odd $p$), the question was answered by Chen, Handschy and Lerman in \cite{Chen2018} (and is also valid in the presence of a magnetic field). They further make use of this fact that those models exhibit a very strong form of the Multiple Peaks property. The goal of this note is to give a rather compact proof of Chaos for any mixed $p$-spin models:

\begin{theorem} \label{thm:main}
There exists $\eps_N \to 0$ such that \eqref{eq:chaos} holds true.
\end{theorem}

Let us point out one consequence of our main theorem. Consider the ground energy function $f_N:\mathcal{H} \to \RR$ defined by
$$
f_N(x) := \max_{\sigma \in \DC} H_N(\sigma; x).
$$
It is easily checked that the function $f$ is $O(\sqrt{N})$-Lipschitz. Thus, due to classical concentration estimates, one has that
$$
\Var[ f_N(\gg) ] = O(N).
$$
It is natural to ask whether this bound can be improved, namely whether $\Var[ f_N(g) ] = o(N)$. This property is often referred to as \emph{superconcentration}, a term coined by Chatterjee, and as shown in \cite[Theorem 1.8]{chatterjee2008chaos}, such a bound is in fact equivalent to chaos. Due to this equivalence, we obtain:

\begin{corollary} 
One has $\Var[f_N(\gg)] = o(N)$ as $N \to \infty$.
\end{corollary}

\subsection{History and related work}

In the physics literature, the study of spin glasses was initiated in the work of Edwards and Anderson \cite{Edwards_1975} and the model with which we are concerned was introduced soon thereafter, by Sherrington and Kirkpatrick \cite{SK75}. In the following years, this line of research has inspired numerous new methods in physics, including Parisi’s replica method. We refer to \cite{MPV86} for a survey of these methods.

In the past two decades, mathematicians have finally managed to start catching up by proving rigorous counterparts to some of the predictions given in the physics literature. Perhaps the most significant breakthroughs are in the works of Talagrand \cite{Talagrand-Parisi} who rigorously established the Parisi formula (building on Guerra's work \cite{Guerra03}), and extended by Panchenko \cite{Panchenko}. Some notable recent works by Subag \cite{Subag17,subag2018free}, and Subag-Panchenko-Chen \cite{SPC18,SPC19} study the geometry of pure states in $p$-spin models.

The phenomenon of chaos for the Sherrington-Kirkpatrick model was first proposed by Bray and Moore \cite{BM87-SK}; A related suggestion was made in the earlier work \cite{MBK82-SK}, and the work \cite{FH86-SK} discusses chaos in the context of a slightly different model. 

The first rigorous result in the direction of chaos is due to Chatterjee \cite{chatterjee2009disorder}, where a "positive temperature" version of chaos was proved. Roughly speaking, Chatterjee shows that for any $\beta > 0$, if $\sigma$ is sampled from the Gibbs measure with finite temperature $\beta$, and another sample $\sigma'$ is taken from the Gibbs measure which corresponds to a small perturbation of the interaction matrix, then $\sigma$ and $\sigma'$ are almost orthogonal with high probability. As a corollary, Chatterjee deduces that the S-K model exhibits the "multiple peaks" property, which roughly refers to the existence of a large number of "close-competitors" to the maximum which are almost orthogonal to each other. This result was generalized in \cite{chen2013disorder} to the case where a magnetic field is present. Several related forms of chaos for similar models such as $p$-spin models appear in \cite{chen2013approach}.

Chaos for the ground state was fully proved for the SK model as well as all even mixed $p$-spin in \cite{Chen2018}.  Their result is valid in the more general case that a magnetic field is also present. A related form of chaos is \emph{temperature chaos} in which the temperature is perturbed rather than the coefficients, see \cite{chen2014chaos, arous2018geometry} and references therein. We refer to \cite{Chen2018} for a more comprehensive review of the related literature.

\subsection{Proof sketch}

Let us discuss some of the ideas and central steps of the proof. The proof is generally "low-tech" in the sense that it doesn't directly use the replica method and the recently developed techniques in spin glass theory, but rather the argument is in the spirit of more classical concentration bounds on Gaussian space. However, one crucial ingredient (Lemma \ref{lem:decay} below) does rely, essentially as a black box, on the Parisi formula.

The first step of the proof is to show, using ideas that essentially appear in Chatterjee's works \cite{chatterjee2008chaos,chatterjee2009disorder}, that it is enough to prove that 
$$
\lim_{N \to \infty} \EE \left [ \left |\frac{1}{N} \left \langle \sigma^*(\gg), \sigma^*(\gg^{\alpha}) \right \rangle \right | \right ] = 0
$$
for some fixed $\alpha > 0$. In other words, chaos under small noise follows from chaos under fixed, positive noise. This is attained by showing that for fixed $N$, the above expression is essentially log-convex with respect to $\alpha$.

So our main goal is to obtain an upper bound on the probability that $\sigma^*(\gg)$ is correlated with $\sigma^*(\gg^\alpha)$, and by symmetry we may assume for instance that $E:=\{\sigma^*(\gg^\alpha) = (1,...,1) \}$ holds true. The Prekopa-Leindler inequality shows that the distribution of $\gg | E$ is log-concave with respect to the Gaussian measure, and a theorem of Harg\'e \cite{Harge04}, shows that such measures are convexly dominated by the respective standard Gaussian translated to have the same barycenter as $\gg | E$. 

In order to show that it is unlikely that $\sigma^*(\gg)$ is correlated with $\One$, we will simply show that with high probability, the maximum of the Hamiltonian over the set $T = \{ \sigma; ~ \frac{1}{N} \langle \sigma, \One \rangle > \eps \}$ is significantly smaller than the expected maximum over the entire hypercube, and that those two maxima are concentrated. The point is now that
$$
\tilde f(x) := \max_{\sigma \in T} H_N(\sigma, x)
$$
is a convex function, therefore, by the convex domination mentioned above, we do not really need to understand the complicated distribution $\tilde f(\gg) | E$, but rather it is enough to replace $\gg | E$ by a translated standard Gaussian. So at this point, it remains to:
\begin{enumerate}
\item
Find an upper bound for $\EE[\tilde f(\gg)]$.
\item 
Show that the translation of $\gg$ towards the barycenter of the event $E$ does not increase the value of $\tilde f(\gg)$ by too much.
\end{enumerate}
By using an argument based on the Gaussian level-1 inequality, it turns out that the translation towards the barycenter increases the expectation by a term at most quadratic in $\eps$. Therefore, it remains to show that the maximum on sections decreases quadratically with the distance of the section from the origin, in other words that
$$
\frac{1}{N} \bigl (\EE f(\gg) - \EE \tilde f(\gg) \bigr ) > c \eps^2.
$$ 
The quadratic addition due to the translation of barycenter can then be eliminated by choosing $\alpha$ to be large enough. The quadratic decay essentially boils down to the differentiability of the Parisi functional with respect to the magnetic field, based on a formula obtained in \cite{AufChen-Parisi}. \\

It should be noted that our theorem only gives an asymptotic result. The reason that we cannot obtain quantitative rates of convergence is due to the fact that we rely on the Parisi formula, for which no explicit rates of convergence are known. In fact, any nonasymptotic version of the Parisi formula will imply a quantitative rate in our result. However, our method of proof is unlikely to produce the \emph{optimal} rates, which are conjectured to be polynomial. This drawback is due to the first step, in which log-convexity is used: It is not hard to see that even a polynomial rate of decay of the correlation for constant noise will only imply a logarithmic improvement for the superconcentration. We point out that Chatterjee's result \cite{chatterjee2009disorder} does imply quantitative (logarithmic) rates of convergence. \\

\textbf{Acknowledgements.} I am grateful to Eliran Subag and Wei-Kuo Chen for a very useful comments and for suggesting a simpler proof for Lemma 4. I would also like to thank Dmitry Panchenko for pointing out to me the reference \cite{Chen2018} shortly after this paper appeared on the Arxiv, and Jian Ding for telling me about this subject back in 2013.

\section{Preliminaries}
Let $\gamma=\gamma_n$ be the standard Gaussian measure on $\RR^n$. Consider the Ornstein-Uhlenbeck semigroup of operators acting on functions $f \in L_2(\gamma_n)$,
$$
P_t[f](x) := \EE_{\Gamma \sim \mathcal{N} (0, \Id)} \Bigl [f \left (e^{-t} x + \sqrt{1-e^{-2t}} \Gamma \right ) \Bigr ]
$$
and its generator $\LL = \Delta - x \cdot \nabla$, so that $\frac{d}{dt} P_t f = \LL P_t f$. We will use two well-known facts regarding the Ornstein-Uhlenbeck semigroup. First, by the commutation relation $\nabla P_t[f] = e^{-t} P_t[\nabla f]$ and by integration by parts,
\begin{equation}\label{eq:byparts}
\frac{d}{dt} \int f(x) P_t [f](x) d \gamma(x) = \int f(x) \LL P_t [f](x) d \gamma(x) = - e^{-t} \int \langle \nabla f(x), P_t[\nabla f](x) \rangle d \gamma(x).
\end{equation}
Second, since the operator $\LL$ is diagonizable in the Hermite basis with integer eigenvalues, there are linear functionals $f \to \alpha_{\ell}(f) \in \left (\RR^N\right )^{\otimes \ell}$, $\ell=0,1,\dots$, such that
\begin{equation}\label{eq:Hermite}
\int f(x) P_t [f](x) d \gamma(x) = \sum_{\ell=0}^\infty ||\alpha_\ell(f)||^2 e^{-\ell t}.
\end{equation}
\subsection{A reduction to chaos for constant noise}
The first step of our proof is to show that, in order to establish \eqref{eq:chaos}, it is enough to show that there exists a constant $\alpha > 0$, which does not depend on $N$, such that 
\begin{equation}\label{eq:nts}
\lim_{N \to \infty} \frac{1}{N} \EE |\langle \sigma^*(\gg), \sigma^*(\gg^\alpha) \rangle| = 0.
\end{equation}
In other words, it is enough to establish that the ground state is chaotic for noise that does not converge to zero with $N$. Recall that $f_N(x) = \max_{\sigma \in \DC} H_N(\sigma; x)$ and remark that, for almost every $x \in \mathcal{H}$, we have $\nabla f_N(x) = \JJ(\sigma^*(x))$, so for almost every $(x,y) \in \mathcal{H} \times \mathcal{H}$, we have
$$
\xi \left (\frac{1}{N} \langle \sigma^*(x), \sigma^*(y) \rangle \right ) = \frac{1}{N} \langle \nabla f_N(x), \nabla f_N(y) \rangle.
$$
Therefore \eqref{eq:chaos} is equivalent to the existence of $\eps_N \to 0$ such that $\lim_{N \to \infty} \varphi_N(\eps_N) = 0$, where
$$
\varphi_N(t) := \frac{1}{N} \EE [ \langle \nabla f_N(\gg), \nabla f_N(\gg^t) \rangle] = \EE \left  [ \xi \left ( \frac{1}{N} \langle \sigma^* (\gg), \sigma^*(\gg^t) \rangle \right ) \right ].
$$
Now, the identities \eqref{eq:byparts} and \eqref{eq:Hermite} imply that
$$
\varphi_N(t) = \frac{1}{N} \int \langle \nabla f_N(x), P_t[\nabla f_N(x)] \rangle d \gamma(x) = \frac{1}{N}  \sum_{\ell=1}^\infty \ell ||\alpha_\ell(f)||^2 e^{-(\ell-1) t}
$$
which implies that $\varphi_N(t)$ is \emph{log-convex}, hence for all $0<s<t$,
$$
\varphi_N(s) \leq \varphi_N(t)^{s/t} \varphi_N(0)^{1-s/t}.
$$
Now, since $\xi(\cdot)$ is continuous and $\xi(0) = 0$, equation \eqref{eq:nts} implies
$$
\lim_{N \to \infty} \varphi_N(\alpha) = \lim_{N \to \infty} \EE \left [\xi \left (\frac{1}{N} \langle \sigma^*(\gg), \sigma^*(\gg^\alpha) \rangle \right ) \right ] = 0.
$$
Finally, remarking that almost surely $\varphi(0) = 1$ and taking $\eps_N = \frac{1}{\sqrt{\log(1/\varphi_N(\alpha))}}$, we attain 
$$
\varphi_N(\eps_N) \leq \varphi_N(\alpha)^{\eps_N/\alpha} \varphi(0)^{1-\eps_N/\alpha} \leq e^{-\frac{1}{\alpha} \sqrt{\log(1/\varphi_N(\alpha))}} \rightarrow 0,
$$
as desired. The rest of the paper is dedicated to proving that equation \eqref{eq:nts} holds true for a suitable choice of $\alpha > 0$.

\subsection{Decay of the maximum on slices} \label{sec:parisi}
Define
$$
M_N(h) = \frac{1}{N} \EE \max_{\sigma \in \DC} \left ( H_N(\sigma; \gg) + h \langle \sigma, \One \rangle \right ).
$$
It turns out that the limit 
$$
M(h) := \lim_{N \to \infty} M_N(h)
$$ 
exists (see \cite{PanchenkoSK}) can be written as the solution of a certain variational problem. Establishing the convergence and computing the limit is a notoriously difficult task, based on deep ideas which have been developed throughout several decades, and was finally accomplished rigorously by Talagrand \cite{Talagrand-Parisi} and Panchenko \cite{Panchenko}. 

For $\eps \in [0,1]$, define
$$
T(\eps) = \left \{ \sigma \in \DC; \tfrac{1}{N} |\langle \sigma, \One \rangle| \in [\eps, 2 \eps] \right \}.
$$
Roughly speaking, we need to establish a quadratic decay, as a function of $\eps$, of the expected maximum of the Hamiltonian on $T_\eps$. This follows as an immediate consequence of the differentiability of $M(h)$ at $h=0$, which follows from a variant of the Parisi formula obtained By Auffinger and Chen in \cite{AufChen-Parisi}.

\begin{lemma} \label{lem:decay}
There exists a constant $c>0$ and a sequence $\eps_N \to 0$ depending only on $\xi(\cdot)$ such that for all $\eps > \eps_N$, one has
$$
\frac{1}{N} \EE \left [  \max_{\sigma \in T(\eps)} H_N(\sigma; \gg) \right ] \leq M_N(0) - c \eps^2.
$$
\end{lemma}
\begin{proof}
As is shown in \cite[Proposition 8]{Chen2018} (using the Parisi formula which appears in \cite{AufChen-Parisi}), the function $M(h)$ is continuously differentiable and $M'(0) = 0$. Consequently, there exists a sequence $\delta_N \to 0$ and a constant $C>0$ (depending only on $\xi$) such that for all $h \in (0,1)$,
$$
M_N(h) \leq M(0) + \delta_N + C h^2.
$$
We therefore have,
\begin{align*}
\frac{1}{N} \EE \max_{\sigma \in T(\eps)} H_N(\sigma; \gg) & \leq \inf_{h \in (0,1)} \frac{1}{N} \EE \max_{\sigma \in \DC} \left ( H_N(\sigma; \gg) + h \langle \sigma, \One \rangle \right ) - \eps h \\
& \leq \inf_{h \in (0,1)} M(0) + \delta_N + C h^2 - \eps h \\ 
& \leq M(0) + \delta_N + \left (\frac{C}{(1+C)^2} - \frac{1}{1+C} \right ) \eps^2 \leq M(0) + \delta_N - c \eps^2
\end{align*}
where $c > 0$ and depends only on $\xi$. The result of the lemma follows.
\end{proof}

\subsection{A convex domination lemma}
At the heart of our argument lies the following lemma, which is obtained by a combination of several classical bounds on Gaussian space.

\begin{lemma} \label{lem:convexdom}
Fix a dimension $n$ and let
$$
(\gg, \gg^t) \sim \mathcal{N} \left (0, \left ( \begin{matrix} \Id & e^{-t} \Id \\ e^{-t} \Id & \Id \end{matrix} \right ) \right ).
$$
Let $\alpha>0$, let $K \subset \RR^n$ be convex and let $\varphi:\RR^n \to \RR$ be convex and $L$-Lipschitz. Then
$$
\PP \Bigl. \Bigl (\varphi(\gg) > \mu + L s \Bigr | \gg^\alpha \in K \Bigr ) \leq 4 e^{-s^2/2}, ~~~ \forall s>0,
$$
where
$$
\mu = \EE \Bigl [\varphi \bigl (\gg + \EE[\gg | \gg^\alpha \in K] \bigr) \Bigr].
$$
\end{lemma}
$$
~
$$
The key to proving the above lemma is the following definition.
\begin{definition}
A random vector $X$ in $\RR^n$ is said to be \emph{log-concave with respect to $\gamma$} if the law of $X$ is of the form $e^{-V} d \gamma$ where $V:\RR^n \to \RR$ is convex.
\end{definition}
\noindent The proof of the lemma follows from the combination of three classical results:
\begin{itemize}
\item 
By the Pr\'ekopa-Leindler inequality, we have that the conditional vector $\gg | \gg^\alpha \in K$ is log-concave with respect to $\gamma$.
\item 
A theorem of Harg\'e (\cite[Theorem 1.1]{Harge04}) asserts that if $X$ is log-concave with respect to $\gamma$, then it is convexly dominated by $\mathcal{N}(\EE[X], \Id)$.
\item 
A generalization of Borel's inequality asserts that Lipcshitz functions evaluated at a random vector which is log-concave with respect to the Gaussian admit sub-Gaussian concentration.
\end{itemize}

We will give a more compact argument via an alternate route, based on the following stochastic construction. Let $X_t$ be distributed as a standard Brownian motion in $\RR^n$ conditioned on the event $X_1 \in K$, adapted to a filtration $\FF_t$. Define $Y_t := \EE[X_1 | \FF_t]$. It is shown, for instance, in \cite[Lemma 13]{EMZ2018clt} that $Y_t$ attains the following property: There exists an $\FF_t$-adapted Brownian motion $W_t$ and an $\FF_t$-adapted matrix-valued process $\sigma_t$ such that $d Y_t = \sigma_t d W_t$ and such that $0 \preceq \sigma_t \preceq \Id$ almost surely for all $t \in [0,1]$.

\begin{proof}[Proof of Lemma \ref{lem:convexdom}]
Let $\Gamma_1,\Gamma_2$ be standard Gaussian random vectors independent of the above processes. Define $\tilde \sigma = \left (\Id - \int_0^1 \sigma_t^2 dt\right )^{1/2}$; remark that this matrix is well-defined since $0 \preceq \sigma_t \preceq \Id$. Define
$$
Z := e^{-\alpha} Y_1 + \sqrt{1-e^{-2\alpha}} \Gamma_1, ~~~~ W := e^{-\alpha} Y_1 + \sqrt{1-e^{-2\alpha}} \Gamma_1 + e^{-\alpha} \tilde \sigma \Gamma_2.
$$
Since $\tilde \sigma^2 + \int_0^1 \sigma_t^2 dt = \Id$, we have that $W \sim \mathcal{N}(e^{-\alpha} \EE[Y_1], \Id)$ (this is justified more carefully in \cite[Proposition 9]{EL14-NormLC}). Remark that $(Z, Y_1)$ has the same distribution as that of $(\gg,\gg^\alpha)$ conditioned on $\gg^\alpha \in K$. Observe that, by the convexity of $\varphi$, we have almost surely
$$
\PP \Bigl . \Bigl ( \varphi(W) \geq \varphi(Z) \Bigr | Z \Bigr )  = \EE \left . \left  [ \PP \Bigl . \Bigl ( \varphi(Z + e^{-\alpha} \tilde \sigma \Gamma_2 ) \geq \varphi(Z) \Bigr | \tilde \sigma, Z \Bigr ) \right | Z \right ] \geq \frac{1}{2}.
$$
Therefore, we have for all $s>0$,
\begin{align*}
\PP \Bigl. \Bigl (\varphi(\gg) \geq \EE[\varphi(W)] + L s \Bigr | \gg^\alpha \in K \Bigl) & = \PP \Bigl (\varphi(Z) \geq \EE[\varphi(W)] + L s \Bigr ) \\
& \leq 2 \PP \Bigl (\varphi(W) \geq \EE[\varphi(W)] + L s \Bigr ) \leq 4 e^{-s^2/2}, 
\end{align*}
where the last equality is due to the Borell-Tsirelson-Sudakov Gaussian concentration.
\end{proof}

\section{Proof of Theorem \ref{thm:main}}
In this section we allow ourselves to omit the subscript $N$ from the notation whenever no confusion is caused. Let $\One=(1,\dots,1) \in \RR^N$. Define 
$$
S := \left \{ x \in \mathcal{H}: \sigma^*(x) = \One \right \} = \left \{ x \in \mathcal{H}: f(x) = H(\One, x) \right \}.
$$
Observe that the set $S$ is convex as it can be defined as the intersection of linear constraints. 

By invariance to the symmetry group of $\DC$, we can write
\begin{align*} 
\EE \bigl [ |\langle \sigma^*(\gg), \sigma^*(\gg^t) \rangle | \bigr ] & = \frac{1}{2^N} \sum_{\sigma \in \DC}  \EE \left [  | \langle \sigma^*(\gg), \sigma \rangle | ~~ | f(\gg^t) = H(\sigma; \gg^t) \right ]  \\
& = \left . \EE \left [  |\langle \sigma^*(\gg), \One \rangle| \right | \gg^t \in S \right ]. 
\end{align*}

In order to prove the theorem, it is enough to establish the existence of $\alpha>0$ and $\delta_N, \delta_N' \to 0$ such that
\begin{equation}\label{eq:nts2}
\PP \left  . \left ( \left | \left \langle \sigma^*(\gg), \One \right \rangle \right | \geq \delta_N N \right | \gg^\alpha \in S \right ) \leq \delta_N'.
\end{equation}
Indeed, since $\frac{1}{N} |\langle \sigma^*(\gg), \One \rangle| \in [-1,1]$, this would imply \eqref{eq:nts} and complete the proof. \\

Recall the definition $T(\eps) = \left \{ \sigma \in \DC; \tfrac{1}{N} |\langle \sigma, \One \rangle| \in [\eps, 2 \eps] \right \}$ from Section \ref{sec:parisi}. Define $E_N := \frac{1}{N} \EE f(\gg)$. The key step in proving \eqref{eq:nts2} is the following lemma.
\begin{lemma} \label{lem:heart}
There exist constants $c,\alpha>0$ and a sequence $\eps_N \to 0$, depending only on $\xi(\cdot)$, such that for every dimension $N$ and all $\eps > \eps_N$,
$$
\PP \left . \left ( \frac{1}{N} \max_{\sigma \in T(\eps)} H(\sigma; \gg) \geq E_N - \eps^2 c ~~ \right | \gg^\alpha \in S \right ) \leq 4 e^{-c \sqrt{N}}.
$$
\end{lemma}

Before we prove this lemma, we will need the following fact. Define $b(S) = \frac{\int_S x d \gamma}{\gamma(S)}$, the Gaussian center of mass of $S$. We then have that,
\begin{lemma}
There exists a constant $C>0$, depending only on $\xi(\cdot)$, such that for all $\sigma \in T(\eps)$, we have 
\begin{equation}\label{eq:com}
\langle b(S), \JJ(\sigma) \rangle \leq C N \eps^2.
\end{equation}
\end{lemma}
\begin{proof}
By symmetry, we have $\PP(\gg \in S) = 2^{-N}$. Therefore, by an application of the level-1 inequality (e.g., \cite[Claim 12]{Eldan-twosided}), we have that $\|b(S)\| \leq \sqrt{(2 \log 2) N}$.

Fix $\alpha \in [-1,1]$ and suppose that $S_\alpha := \left \{\sigma \in \DC; ~ \frac{1}{N} \left \langle \sigma, \One \right \rangle = \alpha \right  \}$ is non-empty. Let $\sigma^1,\sigma^2$ be independently uniformly distributed in $S_\alpha$. It is straightforward to show that
$\frac{1}{N} \langle \sigma^1,\sigma^2 \rangle$ converges in probability to $\alpha^2$. Therefore, 
$$
\frac{1}{\sqrt{N}} \| \EE \JJ(\sigma^1) \| = \sqrt{ \frac{1}{N} \EE \langle \JJ(\sigma^1), \JJ(\sigma^2) \rangle } = \sqrt{ \EE \left [ \xi \left ( \frac{1}{N} \langle \sigma^1, \sigma^2 \rangle \right ) \right ]} \to \sqrt{\xi(\alpha^2)} \leq C \alpha^2
$$
for a constant $C>0$ depending only on the model. By symmetry we also clearly have that $\langle b(S), \sigma \rangle$ is constant over $\sigma \in S_\alpha$. Therefore,
$$
\max_{\sigma \in S_\alpha} \langle b(S), \JJ(\sigma) \rangle = \left \langle b(S), \EE \JJ(\sigma^1) \right \rangle \leq \|b(S)\| \| \EE \JJ(\sigma^1) \| \leq C N \alpha^2,
$$
completing the proof.
\end{proof}

\begin{proof} [Proof of Lemma \ref{lem:heart}]
Define
$$
\tilde f(x) = \frac{1}{N} \max_{\sigma \in T(\eps)} H(\sigma; x) = \frac{1}{N} \max_{\sigma \in T(\eps)} \langle \JJ(\sigma), x \rangle.
$$	
Recall that the set $S$ is convex. Clearly, $\tilde f$ is convex, and it is also easily checked that it is $\frac{C(\xi)}{\sqrt{N}}$-Lipcshitz. Since $\mathrm{span}(\JJ(\sigma))_{\sigma \in \DC}$ is finite dimensional, we can think of $\tilde f$ and $S$ as a convex function and a convex set in a finite dimensional Hilbert space, and thus we can apply Lemma \ref{lem:convexdom}, to obtain
\begin{equation}\label{eq:tildefconc}
\PP \left . \left [\tilde f(\gg) \geq \mu_t + s  \right | \gg^t \in S \right ] \leq 4 e^{-\gamma N s^2}, ~~~ \forall s>0,
\end{equation}
where $\gamma>0$ is a constant depending only on the model and $\mu_t = \EE \left [ \tilde f \left (\gg + \EE[\gg | \gg^t \in S] \right ) \right ]$. Remark that $\EE[\gg | \gg^t \in S] = e^{-t} \EE[\gg | \gg \in S] = e^{-t} b(S)$. According to \eqref{eq:com} and since $\tilde f \left (x+ e^{-t} y\right ) \leq \tilde f(x) + e^{-t} \tilde f(y)$ for all $x,y \in \RR^n$, we have
\begin{align*}
\tilde f \left (\gg + \EE[\gg | \gg^t \in S] \right ) \leq \tilde f(\gg) + e^{-t} \tilde f(b(S)) \leq \tilde f(\gg) + C e^{-t} \eps^{2}.
\end{align*}
Moreover, according to Lemma \ref{lem:decay}, there exists a constant $c>0$ and $\eps_N \to 0$ such that whenever $\eps > \eps_N$,
$$
\EE \left [ \tilde f(\gg) \right ] \leq E_N - c \eps^2.
$$
Thus, by choosing $\alpha$ to be a large enough constant (which does not depend on $N$), we have
\begin{align*}
\mu_\alpha = \EE \left [ \tilde f(\gg + e^{-\alpha} b(S)) \right ] \leq \EE \left [ \tilde f(\gg) \right ] + C e^{-\alpha} \eps^2 \leq E_N - c \eps^2/2.
\end{align*}
Equation \eqref{eq:tildefconc} therefore implies
$$
\PP \left . \left ( \tilde f(\gg) > E_N - c \eps^2/2 + s \right | \gg^\alpha \in S \right ) \leq 4 e^{- \gamma N s^2},
$$
Since we may legitimately assume that $\eps > N^{-0.1}$, taking $s = c \eps^2/4$ concludes the lemma.
\end{proof}

Towards proving that \eqref{eq:nts2} holds true, let us define
$$
T := \left \{ \sigma \in \DC; ~ \tfrac{1}{N} \left |\left \langle \sigma, \One \right \rangle \right | \geq \delta_N \right \},
$$
where $\delta_N$ is a sequence converging to $0$ slowly enough, which we will choose later on. Moreover, let $\alpha,c$ be the constants provided by the above lemma. Our goal is to show that
$$
\lim_{N \to \infty} \PP \left  . \left (  \sigma^*(\gg) \in T \right | \gg^t \in S \right ) = 0.
$$
We may now write $\left \{ \sigma^*(\gg) \in T \right \} \subset A_N \cup B_N$, where
$$
A_N := \left \{ \frac{1}{N} f(\gg) \leq E_N - c \delta_N^2 \right \} ~~~~\mbox{and}~~~~ B_N := \left \{ \frac{1}{N} \max_{\sigma \in T} H(\sigma; \gg) \geq E_N - c \delta_N^2 \right \}.
$$
Since $x \to \frac{f(x)}{N}$ is $\frac{1}{\sqrt{N}}$-Lipschitz, Gaussian concentration gives that for all $s > 0$, we have
$$
\PP \left ( \left |\frac{f(\gg)}{N} - E_N \right | > \frac{s}{\sqrt{N}} \right ) < 2 e^{-s^2/2}.
$$
Now, by symmetry, we have that $f(\gg) | \gg^\alpha \in S$ has the same distribution as $f(\gg)$. Therefore, 
$$
\PP(A_N | \gg^\alpha \in S) = \PP(A_N) \to 0,
$$
as long as $\sqrt{N} \delta_N^2 \to + \infty$. To bound the probability of the second event, write
$$
B_N \subset \bigcup_{i=0}^{\lceil \log_2(1/\delta_N) \rceil} \left \{ \frac{1}{N} \max_{\sigma \in T(2^i \delta_N)} H(\sigma; \gg) \geq E_N - c \delta_N^2 \right \}.
$$
An application of Lemma \ref{lem:heart} gives that as long as $\delta_N > \eps_N$ (with $\eps_N$ being the sequence provided by the lemma), we have
$$
\PP \left . \left ( \frac{1}{N} \max_{\sigma \in T(2^i \delta_N)} H(\sigma; \gg) \geq E_N - \delta_N^2 c \right | \gg^\alpha \in S \right ) \leq 4 e^{-c \sqrt{N}}.
$$
A union bound finally gives
$$
\PP(B_N | \gg^\alpha \in S) \leq 4 \lceil \log_2(1/\delta_N) \rceil e^{-c \sqrt{N}}
$$
thus, choosing $\delta_N = \max(\eps_N, N^{-1/5})$ gives 
$$
\lim_{N \to \infty} \PP(A_N | \gg^\alpha \in S) = \lim_{N \to \infty} \PP(B_N | \gg^\alpha \in S) = 0,
$$
establishing \eqref{eq:nts2}. This completes the proof of the theorem.

\bibliographystyle{alpha}
\bibliography{bib}

\end{document}